\documentclass[]{article}

\usepackage{amsmath,amssymb,amsfonts}
\usepackage{graphicx}
\usepackage{color}

\usepackage{authblk}
\usepackage{amsthm}

\usepackage{tikz}
\usetikzlibrary{decorations.pathreplacing}

\usepackage[english]{babel}
\usepackage[utf8]{inputenc}

\usepackage{cite}

\usepackage{appendix}
\usepackage{url}

\usepackage{geometry}
\newtheorem{theorem}{Theorem}

\newtheorem{lemma}{Lemma}
\newtheorem{definition}{Definition}

\newcommand{\nn}{\nonumber}

\def\a{\alpha}
\def\b{\beta}

\def\d{\delta}

\def\la{\lambda}
\def\m{\mu}
\def\n{\nu}

\def\r{\rho}

\def\s{\sigma}

\def\t{\tau}

\def\p{\psi}

\def\La{\Lambda}
\def\P{\Psi}

\def\<{\langle}
\def\>{\rangle}

\def\Tr{{\rm Tr}}

\def\mN{{\mathcal{N}}}
\def\mH{{\mathcal{H}}}
\def\bI{{\mathbb{I}}}

\def\ha{{\hat{a}}}

\begin{document}

	\title{Taming identical particles for discerning the genuine nonlocality}
	
	\author[1]{Seungbeom Chin \thanks{sbthesy@gmail.com}}
	\author[1]{Jung-Hoon Chun}
\affil[1]{Department of Electrical and Computer Engineering, Sungkyunkwan University, Suwon 16419, Korea}

\maketitle

	\begin{abstract}
This work provides a comprehensive approach to analyze the entanglement between subsystems generated by identical particles, based on the symmetric/exterior algebra (SEA) and microcausality.
Our method amends the no-labeling approach (NLA) to quantify any type of identical particles' entanglement, especially fermions with the parity superselection rule.
We can analyze the non-local properties of identical particles' states in a fundamentally equivalent way to those for non-identical particles, which is achieved by the factorizability of the total Hilbert space of identical particles.
This formal correspondence between identical and non-identical particle systems turns out to be useful for quantifying the non-locality generated by identical particles, such as the maximal CHSH inequality violation and the GHJW theorem of identical particles. 
\end{abstract}

\section{Introduction} \label{intro}

The particle identity is one of the essential quantum features, which is formally represented as the exchange symmetry between states of particles in the first quantization language.
Nevertheless, 
the notion of entanglement mostly has developed based on the presupposition that non-identical particles distribute over distinguishable sectors.
This is mostly by a puzzle that arises when considering the entanglement and particle identity simultaneously~\cite{peres2006quantum}. The exchange symmetry of identical particles evokes superposed forms of wave functions, which seems mathematically equivalent to entangled states. Hence, it becomes subtle to discriminate physically extractable entanglement in multipartite systems of identical particles.

There have been several attempts to quantify the physical entanglement of identical particles ~\cite{schliemann2001j,schliemann2001quantum, ghirardi2002entanglement, benatti2011entanglement,benatti2012bipartite, balachandran2013entanglement,balachandran2013algebraic,franco2016quantum,franco2018indistinguishability, compagno2018dealing, chin2019entanglement} that employ various techniques and definitions, such as Slater number~\cite{ghirardi2002entanglement} or subalgebra restriction~\cite{balachandran2013entanglement,balachandran2013algebraic}.  Also, a recently introduced method, which the authors named no-labeling approach (NLA)~\cite{franco2016quantum,franco2018indistinguishability,compagno2018dealing}, introduced a seemingly unorthodox formalism to compute genuine non-local properties of identical particles. NLA has drawn strong attention since it can handily quantify the entanglement of identical particles with well-known measures such as entropy, concurrence~\cite{franco2016quantum,franco2018indistinguishability} and 
Schmidt number~\cite{sciara2017universality}.

Even if not manifestly mentioned in the above works,  NLA can be understood as the description of identical particles with symmetric and exterior algebras (abbreviated as SEA in this work for  convenience) in the Fock space (see, e.g.,~\cite{fock}). However, as we will discuss in this work,  when NLA is rigorously discussed in the formalism of SEA, one can find that the algebraic relations defined in NLA are incomplete and even incorrect to explain the nonlocality of identical particles in general, especially when the particles condensate in each subsystem. 

Here, we provide an extensive algebraic definition to quantify the nonlocality of identical particles in SEA for the most genearl case. This linear algebra, which we name the \emph{local inner product}, is obtained by imposing the restrictions from microcausality to the Fock space of identical particles.
By defining the local inner product, we can present rigorous separability conditions of identical particles in general and also obtain the rigorous definition for the partial trace of identical particles. As a representative case, we analyze the entanglement of bipartite fermions that preserves the parity superselection rule (PSSR). By evaluating the generation process of entanglement from spatially coherent identical particles, we show that \emph{fermions extract a larger amount of entanglement from the same coherence than bosons do}, which is by the exchange antisymmetry that represents the exclusion property of fermions.

Another crucial property that we can achieve from our approach is that \emph{the total Hilbert space of identical particles is factorizable as non-identical ones}, which is proved with the concept of quantum causality~\cite{navascues2012physical}. Many works that deals with many particle systems in the second quantization language (2QL) assume this factorizability (see, e.g.,~\cite{schuch2004nonlocal,schuch2004quantum}), which seems still puzzling from the viewpoint of Fock space. Here we rigorously resolve this problem. With the Hilbert space factorizability, one can eventually deal with the states of identical particles in the equivalent manner to that of nonidentical particles. As an example, we show that Gisin–Hughston–Jozsa–Wootters (GHJW)  theorem~\cite{gisin1984,hughston1993complete} can be proved to hold directly in the formalism.  The maximal violation of the Clauser-Horne-Shimony-Holt (CHSH) inequality from the superposition of vacuum and four-fermion state is also computed.

Our work is organized as follows: Section~\ref{prel} presents theoretical preliminaries for the quantification of identical particle nonlocality. The concept of the local inner product is defined. Section~\ref{ferm} applies the results of Section~\ref{prel} to the multipartite fermion systems. Section~\ref{factorization} proves the total Hilbert space of identical particles is factorized according to the locality of subsystems and shows that a GHJW theorem and CHSH inequality violation can be obtained for identical particles with the factorization. Section~\ref{discussions} summarizes our results and discusses some possible future works.

\section{Description of identical particles with symmetric/exterior algebra (SEA) and the local inner product}\label{prel}

In this section, we provide mathematical and physical preliminaries for quantifying the entanglement of identical particles. We amend the algebraic relations defined in NLA for analyzing the entanglement of identical particles in general. We first explain the SEA formalism and the role of microcausality in the entanglement of identical particles. SEA and microcausality are combined to introduce a linear algebra, which we dub the \emph{local inner product},  to define the partial trace of identical particles that can be applied to the most general cases. 

\subsection{Identical particles in Fock spaces}\label{math}

The second quantization language (2QL) has usually been considered the best way to treat a set of identical particles in many-body quantum systems. However, when it comes to quantum information processing, 2QL needs quite abstract algebraic methods to define the separability that is not directly  related to Hilbert space tensor product structure~\cite{ benatti2011entanglement,benatti2012bipartite,balachandran2013entanglement,balachandran2013algebraic} and does not provide a direct formalism for the entanglement resource theory.

The no-labeling approach (NLA, \cite{franco2016quantum,franco2018indistinguishability,compagno2018dealing}) is introduced as a midway language\footnote{It is a ``midway language'' because, on the one hand, it resembles the character of 1QL to denote states of particles directly, and on the other hand, it also resembles the character of 2QL to inherently discard particle pseudo-labels.}  to analyze the entanglement of identical particles in a more concrete and intuitive manner. From the mathematical viewpoint, NLA is based on the SEA formalism of identical particles~\cite{fock}. Unlike the first quantization language (1QL), SEA directly reveals the exchange symmetry among particles.
1QL achieves the exchange symmetries of $N$ identical particles by superposing $N$-particle wave functions so that they become symmetric (for bosons) or antisymmetric (for fermions) under the switches of particle pseudo-labels. On the other hand, SEA includes notations that explicitly denote the symmetries, which are the 
\emph{symmetric product $\vee$ and exterior (or antisymmetric) product $\wedge$}~\cite{fock, compagno2018dealing,chin2019reduced}. 





Here, we summarize a mathematical description of identical particles and some crucial algebraic results in SEA (see Ref.~\cite{fock} for a more detailed explanation). Our discussion reveals that the algebraic definitions in NLA for the partial trace of identical particles need correction for a consistent quantification of the entanglement.

We suppose a particle has the state $\P_i = (\p_i,s_i)$ where $\p_i$ is the spatial wavefunction and $s_i$ contains all the possible internal degrees of freedom. The total wavefunction of $N$ identical particles is expressed with the following definition:
\begin{definition}
 (Symmetric and exterior tensor products) 
For a Hilbert space $\mH$ (dim$\mH=d$) 
and state vectors $\{|\P_i\> \}_{i=1}^N\in \mH$ with $|\P_i\> = \sum_{a=1}^{d}\P_i^a|a\>$ and $\P_i^a\in \mathbb{C}$, the symmetric tensor product $\vee$ is defined as
\begin{align}\label{bosons}
|\P_1,\cdots,\P_N\>&\equiv
 |\P_1\>\vee\cdots \vee |\P_N\> \nn \\
 &= \frac{1}{\mN} \sum_{\s\in S_N} |\P_{\s(1)}\>\otimes \cdots \otimes |\P_{\s(N)}\>
\end{align} where $\mN$ is the normalization factor and $S_N$ is the $N$ permutation group. And the exterior (or antisymmetric) tensor product $\wedge$ is defined as
\begin{align}\label{fermion}
|\P_1,\cdots,\P_N\> &\equiv
 |\P_1\>\wedge\cdots \wedge |\P_N\>\nn \\
 &= \frac{1}{\mN} \sum_{\s\in S_N} (-1)^\s |\P_{\s(1)}\>\otimes \cdots \otimes |\P_{\s(N)}\>
\end{align} where $(-1)^\s$ is the signature of $\s$. 
\end{definition}
Then, Eqs.~\eqref{bosons} and \eqref{fermion} correspond to the wavefunctions of $N$ bosons and fermions respectively~\cite{fock}. 
One can notice that the above definition directly connects states written in 1QL to those in SEA (see also Ref.~\cite{chin2019reduced}). 
A closed subspace of $\mH^{\otimes N}$ generated by $|\P_1\>\vee\cdots \vee |\P_N\> $ is denoted by $\mH^{\vee N}$
in which $N$ bosons reside, and a closed subspace of $\mH^{\otimes N}$ generated by $|\P_1\>\wedge\cdots \wedge |\P_N\> $ is denoted by $\mH^{\wedge N}$
in which $N$ femions reside. These two subspaces compose
\emph{Fock spaces}, which are algebraic constructions of single Hilbert space for unfixed number of identical particles.
 The bosonic Fock space over $\mH$ is defined as
$F_b(\mH) =\bigoplus_{N=0}^\infty \mH^{\vee N}$ and 
 the fermionic Fock space as $F_f(\mH) =\bigoplus_{N=0}^\infty\mH^{\wedge N}$ with the definition $\mH^0=\mathbb{C}$.
 $\mH^0$ is the Hilbert space  for \emph{the vacuum state} $|vac\>$. We will see that $|vac\>$ plays an important role in the definition of PSSR entanglement for fermions. 
From now on, ``$\otimes_{\pm}$'' will be used when the algebra can be any of the symmetric and exterior tensor products.

Creation and annihilation operators $(\ha^\dagger, \ha)$ are defined in SEA as follows~\cite{fock}:
$ $\\

\begin{definition}\label{creationannihilation}
 
  The creation operator $\ha^\dagger(\P)$ from $\mH^{\otimes_{\pm} N}$ to $\mH^{\otimes_{\pm} (N+1)}$ is defined as
  \begin{align}\label{creat}
      \ha^\dagger_\P (|\P_1\>\otimes_\pm\cdots\otimes_\pm  |\P_N\>) = |\P\>\otimes_\pm|\P_1\>\otimes_\pm\cdots \otimes_\pm  |\P_N\>. 
  \end{align}
  The annihilation operator $\ha_\P$ from $\mH^{\otimes_\pm N}$ to $\mH^{\otimes_\pm (N-1)}$ is defined with the concept of the interior product $\cdot$ as
  \begin{align}\label{annih}
     \ha_\P (|\P_1\>\otimes_\pm\cdots\otimes_\pm |\P_N\>) &\equiv \<\P|\cdot  |\P_1\>\otimes_\pm \cdots\otimes_\pm |\P_N\> \nn \\
      &=\sum_{i=1}^N (\pm 1)^{i-1} \<\P|\P_i\>|\P_1\>\otimes_\pm \cdots \otimes_\pm (|\P_i\>) \otimes_\pm\cdots \otimes_\pm |\P_N\>,
  \end{align}  where $(|\P_i\>)$ in the last line means that the state $|\P_i\>$ is absent.  
  
The relations of the above operators to the vacuum state $|vac\>$ are defined as 
\begin{align}\label{vacuum}
    \ha_\P|vac\> = 0, \quad \ha^\dagger_\P|vac\> = |\P\>. 
\end{align}
\end{definition}
  
Note that the commutation relations of creation and annihilation operators are given by
\begin{align}
    [\ha_\P,\ha_\Phi]_{\pm} =0,\quad  [\ha_\P,\ha^\dagger_\Phi]_{\pm} =\<\P|\Phi\>
\end{align} where $+$ ($-$) is the commutator (anticommutator) for bosons (fermions). 
  
It is direct to generalize Eqs.~\eqref{creat} and \eqref{annih} to multi-particle creation and annihilation. For example,
\begin{align}
    \ha^\dagger_{\P'}\ha^\dagger_\P (|\P_1\>\otimes_\pm\cdots\otimes_\pm  |\P_N\>) 
    &= \ha^\dagger_{\P'} (|\P\>\otimes_\pm |\P_1\>\otimes_\pm\cdots\otimes_\pm  |\P_N\>) \nn  \\
    &=|\P'\>\otimes_\pm  |\P\>\otimes_\pm |\P_1\>\otimes_\pm\cdots\otimes_\pm  |\P_N\>,
\end{align} etc.
The $n$ particle creation and annihilation processes correspond to the  the following bilinear maps:
\begin{align}
 &\mH^{\otimes_\pm N}\times \mH^{\otimes_\pm n} \to \mH^{\otimes_\pm (N+n)} \quad(\textrm{creation}) \nn \\
 &\mH^{\otimes_\pm N}\cdot  \mH^{\otimes_\pm n} \to \mH^{\otimes_\pm (N-n)} \quad(\textrm{annihilation})
\end{align} where $\times$ and $\cdot$ denote the outer and interior products. 

Then the transition amplitude from a state $|\P_1,\cdots,\P_N\>$ to $|\Phi_1,\cdots,\Phi_N\>$ can be derived from Definition~\ref{creationannihilation} as 
\begin{align}\label{transition}
\<\Phi_1,\cdots,\Phi_N|\P_1,\cdots,\P_N\> 
&=
 \<\Phi_N|\otimes_\pm \cdots \otimes_\pm \<\Phi_1|\cdot|\P_1 \>\otimes_\pm \cdots \otimes_\pm |\P_N\>\nn \\
&=
\left\{
  \begin{array}{@{}ll@{}}
    \frac{1}{\mN^2}Per[\<\Phi_i|\P_j\>] & \text{for bosons} \\
    \frac{1}{\mN^2}Det[\<\Phi_i|\P_j\>] & \text{for fermions}
  \end{array}\right. 
\end{align}
where $Per$ and $Det$ mean the permanent and determinant of a $N\times N$ matrix with entries $\<\Phi_i|\P_j\>$. 

It is worth emphasizing the difference between the LHS and the RHS of the first equality in Eq.~\eqref{transition}.
The LHS denotes
the transition amplitudes of the states that can be rewritten as $\<vac|\ha_{\Phi_N}\cdots\ha_{\Phi_1}$ $ \ha^\dagger_{\P_1}\cdots \ha^\dagger_{\P_N}  |vac\>$, and the RHS denotes the interior product of the corresponding multilinear tensors. They are equal only when two states have the same number of particles. As a simple example to show that the two algebras are not identical in general, we consider $(N,M)=(1,2)$. Then
\begin{align}
    &\<\Phi_1|\P_1,\P_2\> = \<vac|\ha_{\Phi_1} \ha^\dagger_{\P_1}\ha^\dagger_{\P_2} |vac\>=0,
\end{align}
while
\begin{align}
\<\Phi_1|\cdot|\P_1 \>\otimes_\pm |\P_M\>=\<\Phi_1|\P_1\>|\P_2\> \pm \<\Phi_1|\P_2\>|\P_1\>,
\end{align} which is not zero in general.

One can find the same $dot$ ($\cdot$) product definition  with Definition~\ref{creationannihilation} in NLA~\cite{compagno2018dealing} (without mentioning the role of vacuum states). This algebraic identity supports that NLA is naturally obtained from 1QL with the mathematical formularization of SEA. While Definition~\ref{creationannihilation} just presents the mathematical concept of creation and annihilation operators with SEA in the Fock space, a dot product in NLA is introduced to define the partial trace of identical particles, i.e., a physical operation (the definition of identical particle partial trace in NLA is given in Section~\ref{SPT} of this work). However, the identification of partial trace to the interior product of SEA works only when the particles are bosons and obey the particle number superselection rule (NSSR), which we will discuss in Section~\ref{SPT} after explaining the role of microcausality in nonlocality.



\subsection{Microcausality and nonlocality}\label{microc}

To obtain an operational framework for utilizable entanglement from identical particles, the concept of spatially localized operations and classical communications (sLOCC) is introduced in Ref.~\cite{franco2018indistinguishability,compagno2018dealing}, which states that the operations for sLOCC occur at restricted spatial regions.
Considering that the spatial regions correspond to local detectors (modes), this description implies that the extracted entanglement of identical particles within the sLOCC framework is the mode entanglement (see Appendix~\ref{mode}). This type of entanglement generation process is similar to the detector-level entanglement introduced in Ref.~\cite{tichy2013entanglement}. Appendix~\ref{particletomode} reviews the concept of particle and mode entanglements of identical particles. 

A principal prerequisite to specify the entanglement of identical particles is to clarify the particles' spatial relation to modes. Therefore, it presumes a spatial computational basis for quantifying entanglement, $\{|X_a\>\}_{a=1}^{P\ge 2} $, where $X_a$ denote individual subsystems according to their spatial location ($\<X_a|X_b\>=\d_{ab}$).
Then any spatial wavefunction $|\p_i\>$ is expressed as $|\p_i\> = \sum_a\p_i^a|X_a\>$.

With this restriction, operators acting on a subsystems $X_a$ is expressed as
$O_{X_a} = \sum_{r,s}O_{X_a}^{rs}|X_a,r\>\<X_a,s|$ ($r$ and $s$ denotes the possible internal degrees of freedom), which directly satisfies 
\begin{align}\label{commutation}
    [\mathcal{O}_{X_a}, \mathcal{O}_{X_b}] = 0, \qquad (a\neq b).
\end{align}
The above commutation relation connects the entanglement of identical particles extracted under sLOCC to those defined based on algebraic methods~\cite{balachandran2013entanglement,balachandran2013algebraic,benatti2012bipartite,benatti2011entanglement} (see, e.g., Proposition 1 of Ref.~\cite{benatti2012bipartite}).

The commutation relation Eq.~\eqref{commutation} can be understood as a three-dimensional version of \emph{microcausality}, which means that operators acting on spacelike separated regions commute,
\begin{align}
 [\mathcal{O}_{X_a}(t_a),\mathcal{O}_{X_b}(t_b)] =0 \quad \textrm{if} \quad (X_a-X_b)^2-(t_a-t_b)^2 < 0.
\end{align} Hence, sLOCC implies microcausality with given computational basis according to the spatial distribution of subsystems (detectors, modes). And the \emph{separability conditions of identical particles} emerge from the fixation of a spatial computational basis and the presumption of microcausality. Since the separablity condition of identical particles is not explicitly presented in Refs.~\cite{balachandran2013entanglement,balachandran2013algebraic,benatti2012bipartite,benatti2011entanglement}, we provide that for bosons under the restriction of particle number superselection rule (NSSR)~\cite{wick1952intrinsic,wick1970superselection} in Appendix~\ref{NSSR}.

Another crucial role of microcausality is that the partity superselection rule (PSSR) for fermions~\cite{friis2016reasonable,johansson2016comment,gigena2015entanglement, gigena2017bipartite, amosov2017spectral} is derived from the microcausality, which indicates that the entanglement of identical particles under the restiction of SSR presumes microcausality. 


\subsection{The symmetrized partial trace of identical particles in SEA}\label{SPT}

 
One of important theoretical contributions of NLA to the entanglement of identical particles is to suggest a concrete computational method to obtain the reduced density matrix of identical particle states~\cite{franco2016quantum,franco2018indistinguishability,compagno2018dealing}. In NLA, the definition of a partial trace of a state is based on the interior (dot) product~\cite{compagno2018dealing} given in Definition~\ref{creationannihilation}.
By supposing that a complete basis of $X_a$ is given by $\{|\Phi^a_q\>\}_q$ with  $\<\Phi^a_p|\Phi^a_q\>=\d_{pq}$, the partial trace over a subsystem $X_a$ of a state $\r = \sum_{k}p_k|\P_k\>\<\P_k|$ ($\sum_k p_k=1$)  is defined with the corresponding identity matrix $\bI_{X_a} = \sum_q|\Phi^a_q\>\<\Phi^a_q|$ as
\begin{align}\label{partialnsa}
\Tr_{X_a}(\r) = \Tr_{X_a}(\r\bI_{X_a}) \equiv  \sum_{q,k}p_k\<\Phi^a_q|\cdot|\P_k\>\<\P_k|\cdot|\Phi^a_q\>,    
\end{align} where the operation $\cdot$ denotes the interior product (Eq.~\eqref{annih}). 
 
On the other hand, for a definition of partial trace to be valid, it should satisfy the following conditions: 
\begin{itemize}
    \item C1) When a state of identical particles is local (all particles are in the same subsystem), the partial trace of the state is just the trace and becomes $1$ (a number, not an identity matrix). 
    \item C2) When a pure state is nonlocal and separable, the obtained reduced density matrix becomes a pure state. 
\end{itemize}
And we can directly check that the definition of the partial trace~Eq.~\eqref{partialnsa} does not meet the above conditions in general. Actually, it is a valid definition only when the particles preserve the NSSR restriction. 

First, we show that the definition of partial trace as Eq.~\eqref{partialnsa} meets C1 and C2 when the particles preserve NSSR. 
For this case, the basis of a subsystem $X_a$ is set to be $\{|\Phi_{q(n)}^a\>\}_q$ ($\<\Phi^a_{p(n)}|\Phi^a_{q(n)}\>=\d_{pq}$) where $(n)$ denotes that they are all $n$-particle states.
The identity matrix for $X_a$ is given by $\bI_{X_a}^{(n)} = \sum_q|\Phi_{q(n)}^a\>\<\Phi_{q(n)}^a|$. Then, for an arbitary $n$-particle $X_a$-local state $\sum_r \p^r |\Phi_{r(n)}^a\>$ ($\sum_{r}|\p^r|^2=1$), the trace is given by
\begin{align}
    \sum_{q}\<\Phi^a_{q(n)}|\Big(\sum_{r,s}\p^r\p^{s*}|\Phi^a_{r(n)}\>\< \Phi^a_{r(n)}\Big)|\Phi^a_{q(n)}\> =\sum_{q,r,s}\d_{qr}\d_{rq}\p^r\p^{s*} =1.  
\end{align} By convex roof extention, Eq.~\eqref{partialnsa} also satisfies C1 for mixed local states with NSSR. And, since it is direct to see that the reduced density matrix of an arbitary separable state Eq.~\eqref{bisep.} is pure, C2 is also satisfied. 

Second, we show that Eq.~\eqref{partialnsa} does not meet either C1 or C2 in general (without NSSR). 
Suppose that the bosons have two internal degrees of freedom, $\uparrow$ and $\downarrow$. We consider two simple boson states, $\ha^\dagger_{X_a,\uparrow}\ha^\dagger_{X_a,\downarrow}|vac\> \equiv |\uparrow,\downarrow\>_{X_a}$ and $\frac{1}{\sqrt{2}}(\ha^\dagger_{X_a,\uparrow})^2 \ha^\dagger_{\bar{X}_a,\downarrow}|vac\> \equiv  |\uparrow,\uparrow\>_{X_a}\vee|\downarrow\>_{\bar{X_a}}$ (here $\bar{X}_a$ is the complement system of $X_a$). 
Since NSSR is not preserved now, the trace of a local state $|\uparrow,\downarrow\>_{X_a}$ is given from Eq.~\eqref{partialnsa} by
\begin{align}
    \Tr_{X_a} |\uparrow,\downarrow\>\<\uparrow,\downarrow|_{X_a} =& \<\uparrow|\cdot |\uparrow,\downarrow\>\<\uparrow,\downarrow|\cdot|\uparrow\>_{X_a} + \<\downarrow|\cdot |\uparrow,\downarrow\>\<\uparrow,\downarrow|\cdot|\downarrow\>_{X_a}  + \<\uparrow,\downarrow|\cdot |\uparrow,\downarrow\>\<\uparrow,\downarrow|\cdot |\uparrow,\downarrow\>_{X_a}  \nn \\
    =& |\uparrow\>\<\uparrow|_{X_a} + |\downarrow\>\<\downarrow|_{X_a} + 1,
\end{align} which is a nonsensical result, and does not satisfy C1.
And a similar computation show that the partial trace of a separable state $|\uparrow,\uparrow\>_{X_A}\vee|\downarrow\>_{\bar{X_a}}$ is not pure, i.e., C2 is not satisfied, either.  


    

Therefore, we need to remedy the partial trace definition of identical particles Eq.~\eqref{partialnsa} to obtain the partial trace that satisfies C1 and C2 for all the possible situations. To achieve our goal, we introduce a new type of operation, which we name the \emph{local inner product}:
\begin{definition} The local inner product (denoted as $\circ$) is 
 a linear operation defined as a projection between a local state $|\Phi\>$ on a subsystem $X_a$ and a possibly nonlocal state $|\P\>$ on $\mH$. If $|\P\>$ is written as $|\P\> =\sum_{q}\p_q|\P'^q \>_{X_a}\otimes_{\pm} |\P''^q\>_{\bar{X}_a}$ ($|\P''^a\>_{\bar{X}_a}$ is an arbitrary state on the complementary system $\bar{X}_a$ and $\sum_q|\p_q|^2=1$),  $\<\Phi|\circ|\P\>$ is defined as 
 \begin{align}\label{circ}
   \<\Phi|\circ|\P\> =  \sum_q \<\Phi|\P'^q\>_{X_a} |\P''\>_{\bar{X}_a}. 
 \end{align}
 \label{localinner}
\end{definition}
Since the local inner product is linear, Eq.~\eqref{circ} can be directly extended to the projection of arbitrary nonlocal states. Note that this definition is possible only after the specification of subsystems $\{X_a,\bar{X}_a\}$ for the nonlocality; hence, \emph{the microcausality is prerequisite for defining the local inner product.}

By employing Definition~\ref{localinner}, the partial trace that satisfies C1 and C2 is defined as follows:
\begin{definition}\label{partialsea}
    For the identity matrix $\bI_{X_a} = \sum_{q}|\Phi^a_q\>\<\Phi^a_q|$ of a subsystem $X_a$, the partial trace over $X_a$ for a state $\r = \sum_{k}p_k|\P_k\>\<\P_k|$ ($\sum_k p_k=1$) in $\mH$ is defined as
    \begin{align}\label{generalpartial}
        \Tr_{X_a}(\r) = \sum_{q,k}p_k\<\Phi_q^a|\circ|\P_k\>\<\P_k|\circ|\Phi_q^a\>.
    \end{align} 
\end{definition}
In Eq.~\eqref{circ}, $\<\Phi|\P'\> $ can be $\<\Phi|\cdot|\P'\>$ when the particle numbers of $|\Phi\>$ and $|\P'\>$ are equal. Hence, for bosons with NSSR, the above definition becomes equivalent to   Eq.~\eqref{partialnsa} of NLA.

By Definition~\ref{partialsea}, the trace of an arbitrary local state $|\P\>_{X_a} = \sum_{q}\p_q|\Phi^a_q\>$ ($\sum_q|\p_q|^2=1$) in $X_a$ is given by 
\begin{align}
    \Tr_{X_a}(|\P\>\<\P|) &= \sum_{p,q,r}\p_q\p_r^*\<\Phi_p^a|\circ |\Phi^a_q\>\<\Phi^a_r|\circ|\Phi_p^a\>_{X_a} \nn \\ 
    &= \sum_{p,q,r}\p_q\p_r^*\<\Phi_p^a|\Phi^a_q\>\<\Phi^a_r|\Phi_p^a\>_{X_a} =1,
\end{align} which shows that Definition~\ref{partialsea} satisfies C1.  We can see that C2 is satisfied for Definition~\ref{partialsea} by inserting $|\P'^q \>_{X_a}\otimes_{\pm} |\P''^q\>_{\bar{X}_a}$ into Eq.~\eqref{generalpartial}. Therefore, we see that the partial trace of identical particles as Definition~\ref{partialsea} is suitable for deriving a reduced density matrix of identical particles for the most general case.



\section{Entanglement of fermions}\label{ferm}

As we have briefly explained at the end of Sec. \ref{microc}, the microcausality renders one to conceive the PSSR-preserving entanglement among local regions.
In this section, we investigate the entanglement of fermions with PSSR in bipartite systems. The partial trace technic by Definition~\ref{partialsea} is employed to quantify the entanglement. We show that PSSR permits the fermions to have more entanglement than bosons with NSSR. 

\subsection{The separability conditions of fermions with PSSR}


According to the exclusion principle of fermions, the maximal total fermion number $\max(N)$ is determined by the spatial subsystem number $P$ and the internal degrees of freedom $S$, i.e., $\max(N)=PS$.

Before presenting the separablity condition of fermions for the general case, we first consider the simplest case, i.e., a bipartite spin half fermion system ($\max (N) =4$) with the two spatial subsystems ($X,Y$) and internal spin states ($\uparrow,\downarrow$). By PSSR, we treat even and odd parity states distinctively.
The most general form for a separable total state of even parity $|\P^{even}\>$ is given by
\begin{align}\label{separable_feven}
|\P_{even}^{sep}\> =\a\Big((\sum_{s_1=\uparrow,\downarrow}\a_{s_1}|s_1\>)_X\wedge(\sum_{s_2=\uparrow,\downarrow}\a_{s_2}|s_2\>_Y) \Big) +
    \b \Big( ( p|vac\>_X + q|\uparrow,\downarrow\>_X) \wedge (r|vac\>_Y +s|\uparrow,\downarrow\>_Y) \Big)
\end{align} ($|\a|^2+|\b|^2 = |p|^2+|q|^2=|r|^2+|s|^2 = \sum_{s_1}|\a_{s_1}|^2 = \sum_{s_2}|\a_{s_2}|^2=1 $).
Here the total vacuum state $|vac\>$ is expressed in the local form as $|vac\>_X\wedge|vac\>_Y$ and $|X,\uparrow\>\wedge|X,\downarrow\>\equiv |\uparrow,\downarrow\>_X$, etc. Note that, since PSSR is conserved not only in the total system but also in each subsystem, the two terms in the RHS of Eq.~\eqref{separable_feven} can not superpose from the viewpoint of local observers in $X$ and $Y$.
Similarly, 
an odd fermion state $|\P^{odd}\>$ is separable when it has the form
\begin{align}\label{separable_fodd}
|\P_{odd}^{sep} \>
=& \a\Big[\sum_{s_1}\a_{s_1}|s_1\>_X\wedge\big( p|vac\>_Y + \sum_{s_2,s_3}\m_{s_2s_3}|s_2,s_3\>_Y\big)\Big] 
+ \b  \Big[\big(q|vac\>_X + \sum_{s_4,s_5}\n_{s_4s_5}|s_4,s_5\>_X\big)\wedge \sum_{s_6}\a_{s_6}|s_6\>_Y\Big]
\end{align} ($|\a|^2+|\b|^2 =  \sum_{s_1}|\a_{s_1}|^2 = |p|^2+ \sum_{s_2,s_3}|\m_{s_2s_3}|^2 = |q|^2+ \sum_{s_4,s_5}|\m_{s_4s_5}|^2 =\sum_{s_6}|\a_{s_6}|^2 =1 $).

The generalization to the bipartite system with an arbitrary internal $S$ states ($0,1,\cdots ,S-1$) is straightforward.
A set of fermions that spread over two subsystems $X$ and $Y$ with internal $S$ states are separable if and only if the total state $|\P^{sep}\>$ $(= |\P_{even}^{sep}\> + |\P_{odd}^{sep}\>)$ is given by
\begin{align}\label{evensep}
    |\P_{even}^{sep}\> 
    =& \a\Big( \sum_{k=0}^{[\frac{S-1}{2}]}\sum_{s_1,\cdots,s_{2k+1}}a_{s_1\cdots s_{2k+1}}|s_1,\cdots ,s_{2k+1}\>^X\Big) \wedge \Big( \sum_{k=0}^{[\frac{S-1}{2}]}\sum_{s_1,\cdots,s_{2k}}b_{s_1\cdots s_{2k+1}}|s_1, \cdots ,s_{2k+1}\>^Y \Big) \nn \\
    &+ \b\Big(\sum_{k=0}^{[\frac{S}{2}]}\sum_{s_1,\cdots,s_{2k}}c_{s_1\cdots s_{2k}}|s_1, \cdots,s_{2k}\>^X\Big)\Big] \wedge\Big(\sum_{k=0}^{[\frac{S}{2}]}\sum_{s_1,\cdots,s_{2k}}d_{s_1\cdots s_{2k}}|s_1, \cdots,s_{2k}\>^Y\Big)
\end{align}
($[K]$ for a positive real number $K$ is the the biggest integer among smaller integers than $K$) and
\begin{align}\label{oddsep}
    |\P_{odd}^{sep}\>
    =& \a\Big( \sum_{k=0}^{[\frac{S-1}{2}]}\sum_{s_1,\cdots,s_{2k+1}}a_{s_1\cdots s_{2k+1}}|s_1,\cdots,s_{2k+1}\>_X\Big) \wedge \Big( \sum_{k=0}^{[\frac{S}{2}]}\sum_{s_1,\cdots,s_{2k}}b_{s_1\cdots s_{2k}}|s_1, \cdots, s_{2k}\>_Y \Big)\nn \\
    &+ \b\Big(\sum_{k=0}^{[\frac{S}{2}]}\sum_{s_1,\cdots,s_{2k}}c_{s_1\cdots s_{2k}}|s_1,\cdots,s_{2k}\>_X\Big) \wedge\Big(\sum_{k=0}^{[\frac{S-1}{2}]}\sum_{s_1,\cdots,s_{2k+1}}d_{s_1\cdots s_{2k+1}}|s_1, \cdots,s_{2k+1}\>_Y\Big)
\end{align}
 with the definition
\begin{align}\label{vac}
    \sum_{s_1,\cdots,s_{2k}}a_{s_1\cdots s_{2k}}|s_1,\cdots,s_{2k}\>\Big|_{k=0} =a_0|vac\>
\end{align} (each complex coefficients of wavefunctions in Eqs.~\eqref{evensep},~\eqref{oddsep}, and \eqref{vac} are set to satisfy the normalization conditions). 

In Eq.~\eqref{evensep}, the first line is the exterior product of two even local states, while the second is that of two even local states. In Eq.~\eqref{oddsep}, the first line is the exterior product of an odd local state in $X$ and an even local state in $Y$, while in the second line  the parities of $X$ and $Y$ are reversed. We can directly see that  Eqs.~\eqref{evensep} and \eqref{oddsep} correspond to Eqs.~\eqref{separable_feven} and \eqref{separable_fodd} respectively when $S=2$.


\subsection{Reduced density matrix and von Neumann entropy of fermions}

We can apply the above observation to the case when $N$ fermions spread over space are observed by two distinguishable detectors $X$ and $Y$ (Fig.~\ref{passiveop}),  which is discussed for bosons in Ref.~\cite{franco2018indistinguishability,chin2019entanglement}, i.e., 
\begin{align}\label{nfs}
    |\P\>
    =\wedge_{i=1}^N|\P_i\> = \wedge_{i=1}^N(r_i|X,s_i\>+l_i|Y,s_i\>)
\end{align} with  $|r_i|^2 +|l_i|^2 =1$ for all $i$.
This state is separable when $|\P\>$ is of the form~\eqref{evensep} for an even $N$ or \eqref{oddsep} for an odd $N$. Here we derive the reduced density matrix of the fermions, by which we can compute the \emph{entanglement entropy} of the system. 
\begin{figure}[t]
	\centering
	\includegraphics[width=8cm]{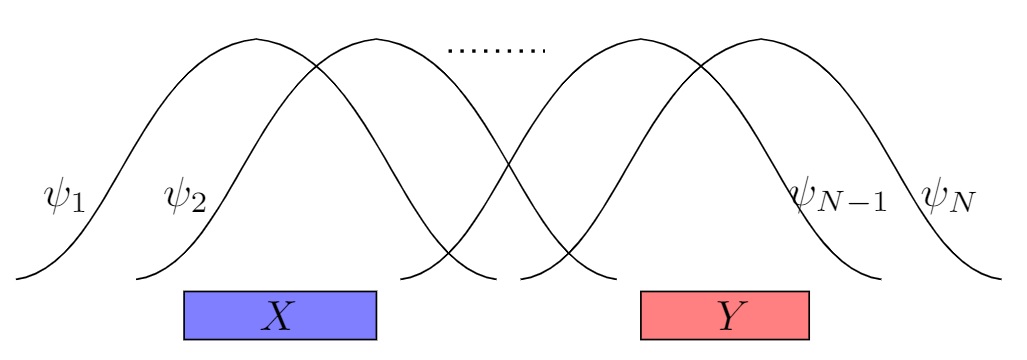}
\caption{$N$ dentical particles projected in a bipartite system. $N$ identical particles with spatial wave functions $(\p_1,\p_2,\cdots ,\p_N)$ are detected by two distinguishable subsystems (modes) $X$ and $Y$, which fix the computational basis as $\{|X\>,|Y\>\}$. }
\label{passiveop}
\end{figure}

We first consider the simplest example, i.e., $N=S=2$ ($s_i =\uparrow$ or $\downarrow$).
For this case, the even and odd identity matrices of the subsystem $X$ are given by
\begin{align}\label{fermionidentity}
    \bI^{even}_X = |vac\>\<vac|_X + |\uparrow,\downarrow\>\<\downarrow,\uparrow|_X, \qquad 
    \bI_X^{odd} = |\uparrow\>\<\uparrow|_X + |\downarrow\>\<\downarrow|_X.
\end{align}
The wave function is written from Eq.~\eqref{nfs} with $N=2$ as
\begin{align}
 |\P\> = &\big( r_1r_2|\uparrow,\downarrow\>_X\wedge|vac\>_Y  +l_1l_2|vac\>_X\wedge|\uparrow,\downarrow\>_Y \big)+ \big( r_1l_2|\uparrow\>_X\wedge|\downarrow\>_Y - l_1r_2|\downarrow\>_X\wedge|\uparrow\>_Y\big)
\end{align} 
By using Eq.~\eqref{fermionidentity} and Definition~\ref{partialsea}, we can obtain two reduced density matrices according to the parity of the subsystem $Y$. The measurable reduced density matrix $\r^{(m)}_Y$ at $Y$ is given by
\begin{align}
 \r^{(m)}_{Y} = p^{even}\r_{Y}^{even} + p^{odd}\r^{odd}_Y
\end{align} where
 \begin{align}
 &\r^{even}_Y 
 = \frac{|r_1r_2|^2 |vac\>\<vac|_Y + |l_1l_2|^2|\uparrow,\downarrow\>\<\downarrow,\uparrow|_Y}{|r_1r_2|^2  + |l_1l_2|^2},\quad 
 \r^{odd}_Y =  \frac{|r_1l_2|^2 |\downarrow\>\<\downarrow|^Y + |l_1r_2|^2 |\uparrow\>\<\uparrow|^Y}{|r_1l_2|^2 + |l_1r_2|^2}, \nn \\
  &p^{even}= |r_1r_2|^2  + |l_1l_2|^2 ,\quad p^{odd} = |r_1l_2|^2 + |l_1r_2|^2. \quad (p^{even} + p^{odd} =1)
  \end{align}

Then the total entanglement entropy $E(\r_Y^{(m)})$, which is defined as
\begin{align}
&E(\r_Y^{(m)})\equiv  p^{even}E(\r_Y^{even}) + p^{odd}E(\r_Y^{odd}), \nn\\ 
 &\qquad (E(\r_Y^{even}) = -\Tr_Y[\r_Y^{even}\log \r_Y^{even}],\quad   E(\r_Y^{odd}) = -\Tr_Y[\r_Y^{odd}\log \r_Y^{odd}])
\end{align} 
is given by
\begin{align}
E(\r^{(m)}_Y)
    =&  -|r_1r_2|^2\log\Big[\frac{|r_1r_2|^2}{|r_1r_2|^2  + |l_1l_2|^2}\Big] -|l_1l_2|^2\log\Big[ \frac{|l_1l_2|^2}{|r_1r_2|^2  + |l_1l_2|^2}\Big] \nn \\
    &-|r_1l_2|^2\log\Big[\frac{|r_1l_2|^2}{|r_1l_2|^2 + |l_1r_2|^2}\Big] -|l_1r_2|^2\log\Big[\frac{|l_1r_2|^2}{|r_1l_2|^2 + |l_1r_2|^2}\Big].
\end{align}
The state is unentangled when one of $(r_1,r_2,l_1,l_2)$ is zero. A noteworthy difference from the bosonic case is that the maximal $E(\r_{ent})$ is given when $|r_1|=|r_2|=|l_1|=|l_2|=\frac{1}{\sqrt{2}}$ by 1, which is twice bigger than maximal $E(\r_{ent})$ for the bosonic case (Fig.~\ref{entropyfm}, compare  Eq.~\eqref{bsent} of Appendix~\ref{NSSR}). This is by the fact that the vacuum state composes the basis of even parity state as seen in Eq.~\eqref{fermionidentity}. We can generalize this quantitative feature to an arbitrary $N$-fermion case Eq.~\eqref{nfs} as the following theorem.
\begin{figure}[t]
    \centering
    \includegraphics[width=6cm]{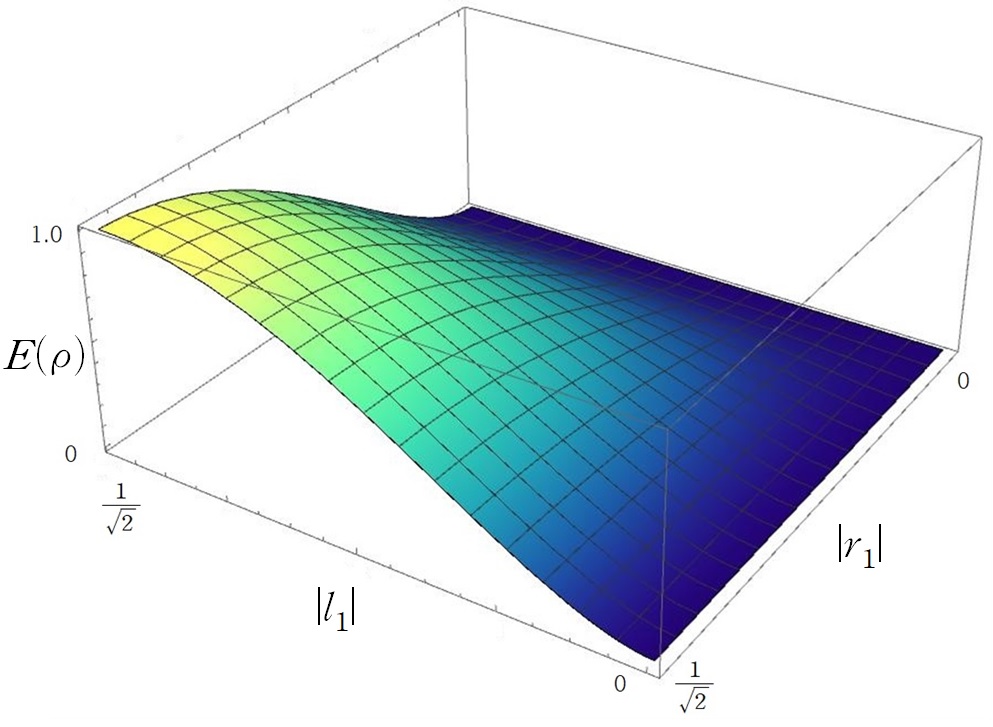}
    \caption{Entanglement entropy of 2 fermions according to the variation of spatial coherence.  The maximal value of the entropy is $1$, which is twice bigger than that of the 2 boson entropy. }
    \label{entropyfm}
\end{figure}


\begin{theorem} When $N \le S$,
the maximal entropy of identical fermions in bipartite subsystems is given  by $N-1$.
\begin{proof}
An $N$ fermion state measured by two detectors at $X$ and $Y$ is expressed as
 \begin{align}\label{nf2d}
     &|\P_1\>\wedge|\P_2\>\wedge\cdots \wedge|\P_N\> \nn \\
     &= \big(l_1|X,s_1\> +r_1|Y,s_1\>\big)\wedge \dots \wedge  \big(l_N|X,s_N\> +r_N|Y,s_N
     \>\big). 
\end{align}
By expanding the above equation according to the particle number per mode, we have 
\begin{align}\label{nf2d1}
&|\P_1\>\wedge|\P_2\>\wedge\cdots \wedge|\P_N\> \nn \\
     &=
     \Big(\prod_{i=1}^N r_i\Big)|vac\>_X\wedge|s_1s_2\cdots s_N\>_Y  
     + \sum_{i}(-1)^{i-1}\frac{l_i}{r_i} \Big(\prod_{j=1}^N r_j\Big) |s_i\>_X\wedge|s_1s_2\cdots(s_i)\cdots s_N\>_Y \nn \\
     &\quad + \sum_{i< j}(-1)^{i+j-1}\frac{l_il_j}{r_ir_j}\Big(\prod_{j}^N r_j\Big)|s_is_j\>_X\wedge |s_1\cdots(s_i)\cdots(s_j)\cdots  s_N\>_Y +\cdots +  \Big(\prod_{j}^N l_j\Big) |s_1s_2\cdots s_N\>_X\wedge|vac\>_Y,
 \end{align}     
 where $(r_i)$ means that it is absent in the ket. By defining 
$(N-n,n)$ as the summation of states with $N-n$ fermions in $X$ and $n$ fermions in $Y$, Eq.~\eqref{nf2d1} is rewritten as
\begin{align}
|\P_1\>\wedge|\P_2\>\wedge\cdots \wedge|\P_N\> &= (N,0) + (N-1,1)+ (N-2,2) + \cdots +(0,N) \nn \\
 &= \Big[ (N,0) + (N-2,2) + (N-4,4) + \cdots \Big] + \Big[ (N-1,1) + (N-3,3)  + \cdots \Big]. 
\end{align}
 By PSSR, we can see that a term in the first bracket of the second line of Eq.~\eqref{nf2d} cannot superpose with a term in the second bracket according to their local parity.

Each $(N-n,n)$ has $\binom{N}{n}$ terms, by which the numbers of the terms in each group are equal.  The state is maximally entangled when $|l_i|=|r_i| =\frac{1}{\sqrt{2}}$ for all $i$, which makes the absolute value of all the amplitudes $2^{-N/2}$. 
By combining all these facts, the PSSR-preserving entanglement entropy of bipartite $N$ fermions is given by
 \begin{align}
     -2\times \frac{1}{2}( \frac{1}{2^{N-1}}\log\Big[\frac{1}{2^{N-1} }\Big]\times 2^{N-1} ) = N-1.
 \end{align}

\end{proof}

\end{theorem}


This monotonic increase of maximal entanglement along the particle number is absent in the bosonic case (see Appendix~\ref{NSSR}), which is because
PSSR permits more fermionic terms to superpose.

The discussion so far has shown that the quantum non-locality of identical particles can be analyzed in a very similar manner to that of non-identical particles with the definition of partial trace (Definition~\ref{partialsea}). In the next section, we will show the factorizability condition of identical particle Hilbert space by microcausality, by which the optimization of identical particle states  is  possible for the quantification of entanglement.


\section{Factorization of Hilbert space}\label{factorization}

Once particles are grouped by their locations, quantifying the physically tangible entanglement of identical particles seems to follow the same process with the non-identical particle case. 
For example, observing the bosonic  state~\eqref{bisep.}, the symmetric tensor product $\vee$ between $|\P_n^X\>$ and $|\P_{N-n}^Y\>$ plays the role of  the direct tensor product $\otimes$ in non-identical particle systems.  Definition~\eqref{localinner} of the local inner product also shows that $\otimes_\pm$ works the same as $\otimes$ under the restriction of microcausality. 

Here we show that this correspondence is not a coincidence and the (anti-) symmetric products $\otimes_{\pm}$ can be replaced with $\otimes$.
In other words, \emph{the Hilbert space of identical particles are factorizable} as that of non-identical particles.
The factorized Hilbert spaces of identical particles are, however, not particle Hilbert spaces but \emph{local} Hilbert spaces, in which each local subsystem corresponds to a Hilbert subspace that constructs the total Hilbert space.
The following theorem clearly states the factorizability of the local Hilbert space.  \\

\begin{theorem}
 If identical particles spread over two subsystems $X$ and $Y$, and the subsystems are spatially distinguishable, then the total Hilbert space is fatorized as  $H_X\otimes H_Y$.   \label{factor} 
\end{theorem}

\begin{proof}

To proof the factorizability of the bipartite state with identical particles, we employ the concept of \emph{quantum causality} introduced in Ref.~\cite{navascues2012physical}. It was shown in the work that a Hilbert space $\mH$ is factorizable into two Hilbert space, i.e., $\mH = \mH_X\otimes \mH_Y$ if and only if the system has the quantum causality (Lemma 4 of Ref.~\cite{navascues2012physical}. Therefore, if the system of identical particles is quantum  causal, then the Hilbert space of the identical particles is factorizable as $\mH = \mH_X\otimes \mH_Y$. 

First, we briefly explain what the quantum causality is. Assume that Xabier is in $X$ and Yoko is in $Y$. Xabier can choose a measurement operation $x$ and produce a datum $q$, and Yoko can choose  $y$ and produce $r$. If they can compare their results after obtaining sufficiently many data, they can estimate the set of probability distributions $\{P(q,r,|x,y)\}$ for all possible $(q,r,x,y)$. Then, the notion of quauntum causality is defined as follows~\cite{navascues2012physical}:\\

\begin{definition}

  $P(q,r|x,y)$ is quantum causal if there exist a Hilbert space $\mH_Y$, projector operators $\{F_r^y:\sum_rF_r^y=\bI_Y\}$, and a set of subnormalized quantum states $\{\s_q^x\}$ (a possible state of Yoko when Xabier activates $x$ and produces $q$) such that 
  \begin{align}\label{quansal}
    & P(q,r|x,y) = \Tr(F_r^y\s_q^x),\quad 
     \sum_q\s_q^x=\s.
  \end{align}
Here $\s$ is independent of $x$.
\end{definition} 
Hence, the statement that Yoko's system is quantum causal means that it is independent of Xabier's system and also compatible with quantum mechanics. 

And, it is not hard to see that the quantum system of identical particle is quantum causal. 
For a given state $|\P\>$ ($\in \mH$),
a subnormalized state of Yoko corresponding to the data $q$ of Xabier is obtained from Definition~\ref{localinner} and \ref{partialsea}, i.e,
\begin{align}\label{subnorm}
 \s_q^x = \<\Phi_q^X|\circ |\P\>\<\P|\circ |\Phi_q^X\>.
\end{align} We know that $\sum_q\s_q^x$ is the partial trace of $|\P\>$ that is independent of the basis choice in $X$. And $P(q,r|x,y)$ is computed from Eq.~\eqref{subnorm} as Eq.~\eqref{quansal}. Thus, the system of identical particles is quantum causal.

In conclusion, as we have mentioned at the beginning of the proof, the Hilbert space of identical particles is factorizable by the quantum causality. 
\end{proof}

With Theorem~\ref{factor}, we can write a $N$-particle state, e.g.,
\begin{align}
|X,s_1\>\otimes_{\pm} \cdots \otimes_{\pm}|X_,s_n\>\otimes_{\pm}|Y,s_{n+1}\>\otimes_{\pm} \cdots \otimes_{\pm}|Y_,s_{N}\> 
\end{align}
in a factorized form
\begin{align}
|s_1,\cdots,s_n\>^X\otimes|s_{n+1}, \cdots,s_{N}\>^Y
\end{align}
(see Fig.~\ref{equiv}).

\begin{figure}[t]
	\centering
	\includegraphics[width=8.5cm]{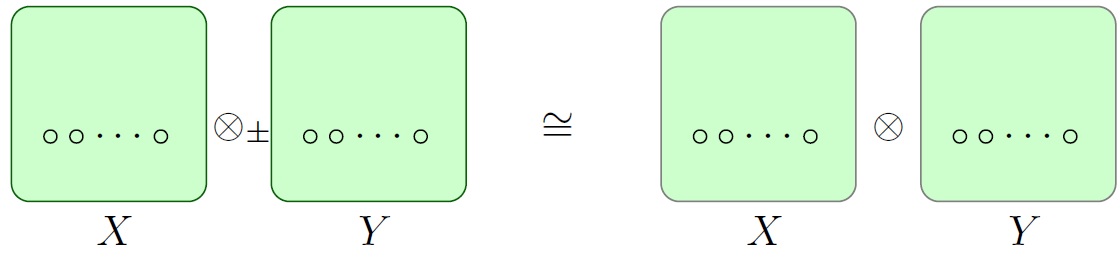}
	\caption{When the subsystems $X$ and $Y$ are distinctive, or more rigorously speaking, spacelike separated, the total Hilbert space $\mH_X\otimes_\pm\mH_Y$ is equivalent to $\mH_X\otimes\mH_Y$. In other words, the total Hilbert space of identical particles is factorizable according to the local distribution of subsystems.} 
\label{equiv}
\end{figure}

Theorem~\ref{factor} is closely related to Tsirelson's theorem~\cite{tsirelson_bellinequalities}, which shows that a quantum system with a factorized Hilbert space is equivalent to a system with two sets of commuting projection operators in a finite-dimensional Hilbert space.
Theorem~\ref{factor} can be applied to the entanglement problems of quantum fields with identical particles.  If each region is supposed to separate far enough from each other, the factorization property of Hilbert space is still valid in quantum fields. On the other hand, if the sub-regions are adjacent to each other, one should cautiously consider the boundary effect. 

The practical advantage of factorizing identical particles' Hilbert spaces is that it makes simpler the derivation of several non-local properties in identical particles' systems. It will become clear by seeing the identical particle version of the GHJW theorem and the CHSH inequality violation in the following discussion.


\subsection{GHJW theorem of identical particles}\label{ghjw}

Here we see how an entangled state of identical particles can raise a nonlocal phenomenon by delving into the GHJW  theorem~\cite{gisin1984,hughston1993complete}, by which any purifications of mixed states on the extended system should have a specific local unitary relation. We show that the theorem is still valid with states of identical particles.
Even if we focus on the bosonic case here, its extension to the fermionic case is straightforward.\\

\begin{lemma}
		Suppose that $|\P\>$ and $|\P'\>$ are $N$-boson vectors in $\mH^{\vee N}$ so that $n$ particles  locate in $X$. If $\Tr_Y|\P\>\<\P| =\Tr_Y|\P'\>\<\P'| $, then there exists a unitary operation in the system $U =\mathbb{I}_X\otimes U_Y$ that satisfies $|\P\> = U|\P'\>$.

\end{lemma}

	\begin{proof}
		The reduced density matrix can be written as
		\begin{align}\label{rdm}
		    \Tr_Y|\P\>\<\P| = \Tr_Y|\P'\>\<\P'|=\sum_{\vec{s}} w_{\vec{s}}|\vec{s}\>\<\vec{s}|_X
		\end{align} where $\vec{s} = (s_1,\cdots,s_{n})$. For any complete orthonormal basis set $\{\vec{r} = (r_1,\cdots ,r_{N-n})\}$ of $Y$, we can write $|\P\>$ as
		\begin{align}\label{psi}
		    |\P\> = \sum_{\vec{s},\vec{r}}\p_{\vec{s},\vec{r}} |\vec{s}\>_X\otimes |\vec{r}\>_Y
		\end{align} using Theorem 1.
		By defining $|\vec{s}\>_Y =\sum_{\vec{r}}\p_{\vec{s},\vec{r}}|\vec{r}\>_Y$, Eq.~\eqref{psi}  is given by
		\begin{align}
		    |\P\> =\sum_{\vec{s}}|\vec{s}\>_X\otimes |\vec{s}\>_Y. 
		\end{align}
		Combining Eqs.~\eqref{rdm} and \eqref{psi}, we obtain
	  $\<\vec{s}|\vec{t}\>_Y=\d_{\vec{s},\vec{t}}w_{\vec{s}}$.
	  Hence, by defining an orthonormal set  
	 $\{|\hat{s}\>_Y\equiv  |\vec{s}\>_Y/\sqrt{w_{\vec{s}} }\}_{\vec{s}}$, $|\P\>$ is finally written as 
	 \begin{align}
	     |\P\> = \sum_{\vec{s}} \sqrt{w_{\vec{s}}}|\vec{s}\>_X\otimes |\hat{s}\>_Y
	 \end{align} The dimension difference of $\mathcal{H}_X$ and $\mathcal{H}_Y$ is not a problem here. When $\textrm{dim}\mathcal{H}_X \neq \textrm{dim}\mathcal{H}_Y$, the number of zero eigenvalues of $\Tr_X|\P\>\<\P|$ and $\Tr_Y|\P\>\<\P|$ differ so that the nonzero eigenvalue numbers become equal.      
     Applying the same process, $|\P'\>$ can be expressed with another orthonormal set $\{\hat{s}'\}_{\vec{s}}$ as
     \begin{align}
      |\P' \> = \sum_{\vec{s}}\sqrt{w_{\vec{s}}} |\vec{s}\>_X\otimes |\hat{s}'\>_Y. 
     \end{align}
 Then the two orthonormal bases  $\{ |\hat{s}\>_Y\}$ and $\{|\hat{s}'\>_Y\}$ are connected by a unitary tranformation $U_Y \equiv \sum_{\vec{s}}|\hat{s}\>\<\hat{s}'|_X$, by which $|\P\>$ and $|\P'\>$ are connected by
 \begin{align}
    |\P\> = (\mathbb{I}_X\otimes U_{Y})|\P'\> \equiv U|\P'\>.
 \end{align}
 
	\end{proof}
	
	Using the above lemma, any $|\P'\>$ that satisfies $\r_X = \Tr_Y|\P\>\<\P|$ can be transformed to $|\P\> = \sum_{\vec{s}}\sqrt{w_a}|\vec{s}\>_X\otimes|\vec{s}\>_Y$, which results in the GHJW theorem of bosons: 

\begin{theorem}
    	 (GHJW theorem for identical particles) Suppose $N$ bosons locate in two orthogonal subsystems $X$ and $Y$ with internal states $s_i$ ($i=1,\cdots, N$). 
		The total state of the bosons $|\P\>$ is a vector in $\mH^{\vee N}\equiv \mH^{\vee n}\otimes \mH^{\vee (N-n)}$ with $\r_{X}=\Tr_Y|\P\>\<\P|$. For any convex summation form of $\r_{X}= \sum_{a}w_{a}|\P^a_{(n)}\>\<\P^a_{(n)}|_X$ ($w_a\ge 0$, $^\forall a$), there exists an orthonormal set $\{|\P^a_{n}\>\}_{a}$  of the subsystem $X$ such that \begin{align}
		|\P\> = \sum_{a}\sqrt{w_a}|\P_{n}^a\>_X\otimes|\P^a_{(N-n)}\>_Y.
		\end{align}
\end{theorem}

This theorem shows that an observer at $Y$ can choose the state of $X$ by performing a measurement and sending the result to an observer at $X$, hence the total system is non-local.

\subsection{Bell inequality violation with bipartite two fermions}\label{bi}

As another exemplary phenomenon of nonlocality that arises from the entanglement of identical particles, we discuss the maximal violation of the Clauser-Horne-Shimony-Holt (CHSH) inequality~\cite{clauser1969proposed}, a standard example of Bell inequalities (BI).
Even though several types of entangled states can generate the violation of the CHSH inequality, here we focus on the bipartite two-level fermionic system and show that a superposition of the vacuum and two fermions can violate the inequality. 
As already mentioned, the factorizability of bipartite systems is used for the manifest verification of the relations.

To discuss BI including the CHSH inequality, we assume the independence of two systems $X$ and $Y$ and a local hidden variable $\la$ ($\in \La$) that determines the probabilities for the subsystems to output data. Thus for any datum $a_x$ from  any measurement operation $x$ of $X$, we obtain a probability function $f_X(a_x,\la)$. By denoting the probablity measure on $\La$ as $M$, the probability distribution is computed as  $P(a_x,b_y|x,y) = \int M(d\la) f_X(a_x,\la)f_Y(b_y,\la)$.   

For the case of CHSH inequality, $a_x$ and $b_y$ have two valued observables, i.e.,  $+1$ and $-1$. Then the mean values of $a_x$ and $b_y$ for a given $\la$ is denoted as $\<a_x(\la)\> = f_X(+1_x,\la) -f_X(-1_x,\la)$ and $\<b(\la)\> = f_Y(+1_x,\la) -f_Y(-1_x,\la)$ ($-1\le\<a_x(\la)\>\le 1$ and $-1\le \<b_y(\la)\>\le 1$). The correlation function is written as $\<a_x,b_y \>= \int M( d\la) \<a_x\>\<b_y\>\r$.
We consider the correlation bound for the two different setting for each detector, i.e, $x=1,2$ and $y=1,2$. The CHSH inequality can be 
derived using the inequality $\<a_{\t_1\t_2}\> \equiv \frac{1}{4}\<(1 +\t_1 a_1)(1 +\t_2 a_2)\>\ge 0$ ($\t_i\in \{+,-\}$) as
\begin{align}
 E(a,b)&\equiv \big|\<a_1,b_1\>+\<a_2,b_1\> + \<a_1,b_2\> -\< a_2,b_2\>\big| \nn \\
 &= \big| \<a_1+a_2, b_1)\> + \<a_1-a_2,b_2\> \big|  \nn \\
&= 2\big| \<a_{++},b_1\>-\<a_{--},b_1\>  + \<a_{+-},b_1\>-\<a_{-+},b_2\> \big| \nn \\
&\le 2\<a_{++}+a_{--} +a_{+-}+a_{-+}, \bI\> =2, 
\end{align} where the last inequality is from the relation $|\<a_{\t_1\t_2},b_j\>| \le \<a_{\t_1\t_2},\mathbb{I}\>$.


For the existence of quantum correlation with Hermitian observables, one can show that the Bell inequality is maximally violated when $ E(a,b) =2\sqrt{2}$ (see, e.g., Ref.\cite{summers1987bell}). This maximal violation can be achieved in the bipartite spin half fermionic system when the fermion state is given by
\begin{align}\label{bellstate}
  |\P^{even}_-\> =& \frac{1}{\sqrt{2}}\big(|vac\>^X\otimes|\uparrow,\downarrow\>^Y -|\uparrow,\downarrow\>^X \otimes|vac\>^Y\big).  
\end{align}
The above state can be obtained from a two fermion state prepared in a system $Z$ 
\begin{align}
    |\P\> &= |\uparrow,\downarrow\>^Z,
\end{align}
which evolves so that the fermions arrive at the systems $X$ and $Y$ in the following form,
\begin{align}
    |\P\>
    &= |\p_1,\uparrow\>\wedge|\p_2,\downarrow\>
\end{align}
with $\p_1 = \frac{1}{\sqrt{2}} (X -Y)$ and  $\p_2 = \frac{1}{\sqrt{2}} (X +Y)$. From PSSR, we can obtain an even-parity fermionic state Eq.~\eqref{bellstate} with probability $1/2$.

Considering that $|vac\>$ and $|\uparrow,\downarrow\>$ are the only two possible independent states per subsystem in this setup (note that the antisymmetric state $|\uparrow,\downarrow\>$ is invariant under any unitary operation), $|\P_-^{even}\>$ is one of fermionic \emph{Bell-like states}.
In this basis,
we can construct three Pauli matrices as follows:
\begin{align}
 \s_1=
 \begin{pmatrix}
 |vac\>\<\uparrow,\downarrow| +|\uparrow,\downarrow\>\<vac|
 \end{pmatrix}, \quad 
 \s_2=
 \begin{pmatrix}
 -i|vac\>\<\uparrow,\downarrow| +i|\uparrow,\downarrow\>\<vac|
 \end{pmatrix}, \quad 
 \s_3=
 \begin{pmatrix}
 |vac\>\<vac| -
|\uparrow,\downarrow\>\<\uparrow,\downarrow|
 \end{pmatrix},
\end{align}
and $\vec{\s}\cdot\hat{n} = \sum_{j=1}^3\s_j\hat{n}_j$ for an arbitrary three-dimensional unit vector $\hat{n}$. 

Then, by setting
\begin{align}
a_1= (\vec{\s}\cdot\hat{n})_X\otimes \bI_Y, \quad a_2= (\vec{\s}\cdot\hat{m})_X\otimes\bI_Y,\quad 
b_1= \bI_X\otimes(\vec{\s}\cdot\hat{n}')_Y, \quad b_2= \bI_X\otimes(\vec{\s}\cdot\hat{m}')_Y
\end{align}
($\otimes$ comes from Theorem 1)
so that the unit vectors $(\vec{n},\vec{m},\vec{n}',\vec{m}')$ satisfy  $\hat{n}\cdot \hat{n}' = \hat{m}\cdot \hat{n}'=\hat{m}\cdot \hat{m}'=-\hat{n}\cdot\hat{m}' =\frac{1}{\sqrt{2}}$,
the maximal Bell inequality violation is obtained, i.e., 
\begin{align}
    |\<\P^{even}_-|(a_1b_1+a_2b_1 + a_2b_2 - a_1b_2)|\P^{even}_-\>|              = 2\sqrt{2}. 
\end{align}


\section{DISCUSSIONS}\label{discussions}

By employing SEA and microcausality,
we have suggested a theoretically rigorous method to quantify any type of identical particles' entanglement, which corrects the algebraic relation for the definition of partial trace in the no-labeling approach (NLA). In this formalism, the total Hilbert space can be factorized according to the location of the particles. In addition, some non-local properties that are seemingly hard to quantify with identical particles, such as the GHJW theorem and BI violation, are handily analyzed. 

Possible applications of our current work are diverse. 
For example, Ref.~\cite{barros2019entangling} theoretically and experimentally verified the quantitative relation of identical particle's entanglement to particle indistinguishability and spatial overlap, in which the partial trace technique based on SEA is used. We expect similar experiments with a larger number of bosons or fermions are possible. It is also an intriguing development to establish a rigorous quantum resource theory of identical particles (see Ref.~\cite{morris2019entanglement} for  related research for the bosonic case) and apply it to more general field-theoretic systems.

\section*{Acknowledgements}
SC is supported by the National Research Foundation of Korea (NRF, NRF-2019R1I1A1A01059964). JC is supported by the Korea Ministry of Trade, Industry and Energy (MOTIE) under Grant 10008040.

\appendix
\numberwithin{equation}{section}

\section{Particle entanglement and mode entanglement}\label{particletomode}

The exchange symmetry and the Hilbert space non-factorizability of identical particles raise the question of specifying subsystems. According to the elements of subsystems that compose the total system, the entanglement of identical particles can be the \emph{particle entanglement} or the \emph{mode entanglement}.  The imposition of different elements to subsystems corresponds to different quantification of entanglement. 
Here we explain the concepts of particle entanglement and mode entanglement and how they are related at the level of detectors.

\subsection{Particle entanglements}\label{particle}

Particle entanglement identifies particles as subsystems. Since each particle is considered a subsystem, this definition implies that \emph{the total system preserves the particle number $N$, with $N$ = (number of particles) = (number of subsystems)}. Here we suppose there are $N$ identical particles in a pure state $|\P_1\>\otimes_\pm |\P_2\>\otimes_\pm \cdots\otimes_\pm |\P_N\>$ (Eqs. \eqref{bosons} and \eqref{fermion}).

What can we say about the entanglement of this state? First of all, one can consider the superposition of the particles originated from the exchange symmetry as an entanglement (Ref.~\cite{killoran2014extracting, cavalcanti2007useful, morris2019entanglement}).
For example, when $N=2$, the identical particle state $ |\P_1\>\otimes_\pm|\P_2\> $ is written in 1QL by
\begin{align}
\label{2particlestate}
    |\P_1\>\otimes_\pm|\P_2\> =\frac{1}{\sqrt{2}} \big(|\P_1\>_A\otimes |\P_2\>_B \pm |\P_2\>_A\otimes |\P_1\>_B\big). 
\end{align} Here $A$ and $B$ are particle labels, which we explicitly write to clarify the particle subsystems. If $A$ and $B$ are considered subsystems, Eq.~\eqref{2particlestate} is an entangled state because it cannot be expressed as a tensor product wave function.
Since no physical detector can address individual particles ($A$ and $B$ are hence called ``pseudolables''), this type of particle entanglement is usually considered artificial entanglement, just dependent on the mathematical form to express identical particles. However, 
Ref.~\cite{killoran2014extracting, cavalcanti2007useful, morris2019entanglement} suggested some protocols to extract this mathematical entanglement into detectable subsystems. These results show that the particle identity is a kind of quantum resource that can be transferred to the mode entanglement.

Nevertheless, it is still true that 
one can discuss the actual entanglement only after discarding this artificial entanglement. Ghiradi et al.~\cite{ghirardi2004general} proposed the concept of \emph{Slater number} for such a discrimination.
According to this criterion, a state is not entangled when it can be expressed as the (anti-) symmetric form under the particle label exchange. For example, a state of Eq.~\eqref{2particlestate} is separable because it is totally (anti-) symmetric. One can directly see that SEA reveals such a property very clearly since every state expressed in SEA inherently discards the superposition of wavefunctions from the exchange symmetry.

\subsection{Mode entanglements}\label{mode}

On the other hand, mode entanglement identifies spatial modes as subsystems. Orthogonal states that compose bases of the local Hilbert spaces is described by the particle number and the possible internal degrees of freedom. 2QL is usually suitable for describing this type of entanglement. Supppose $N$ particles can be found in two independent spatial modes $X$ and $Y$ with no internal degree of freedom. Then a pure state $|(N-n)_X,n_Y\>$ is mode-separable while $\frac{1}{\sqrt{2}}(|(N-n)_X,n_Y\> + |(N-m)_X,m_Y\>)$ ($n\neq m$) is mode-entangled. It should be noted that the criterion for the separability of modes changes when the identical particles follow superselection rules~\cite{wiseman2003entanglement,schuch2004nonlocal,schuch2004quantum,gigena2015entanglement,gigena2017bipartite}.

The mode entanglement is physically extractable entanglement because detectors have access to each mode that is a distinguishable subsystem~\cite{dalton2017quantum}. Moreover, it is proper to state that all the possible genuine entanglements that physical observers can extract are mode entanglements. 

However, the definition of mode entanglement gives rise to a puzzle when single-particle states are considered. 
Let us suppose that a particle can be found in two modes $X$ and $Y$. Then a state
\begin{align}
    \frac{1}{\sqrt{2}}(\ha^\dagger_X|vac\> + \ha\dagger_Y|vac\>) = \frac{1}{\sqrt{2}}(|1_X,0_Y\> +|0_X,1_Y\>), 
\end{align} is mode-entangled, while it is not particle-entangled since it can be written in SEA and 1QL as
\begin{align}
    \frac{1}{\sqrt{2}}(|X\> + |Y\>).
\end{align} Ref.~\cite{wiseman2003entanglement} proposes a method to overcome this confusion with the imposition of the superselection rule to states.

\subsection{Conversion of particle entanglement to mode entanglement}\label{ptom}

An essential property of particle entanglement explained in Sec.~\ref{particle} is that it entirely depends on the formal structure of wave functions, by which
the authors of Ref.~\cite{tichy2013entanglement} called it ``a priori entanglement''. It is pointed out in Ref that this criterion is valid only when each particle is unambiguously assigned to one of detectors, i.e., a particle in $\P_1$ is always observed by the detector $L$ and the other in $\P_2$  by the detectors $R$. 

On the other hand, if $\P_1$ and $\P_2$ are spatially ambiguous, Eq.~\eqref{2particlestate} is no more a definitely separable state. Consider the case when the particles can be observed at both detectors, which is mathematically described as 
\begin{align}
    |\P_i\> = |\p_i,s_i\> = r_i|R, s_i\> +l_i|L,s_i\> \qquad (i=1,2)
\end{align} where $r_i$ and $l_i$ are complex numbers that satisfy $|r_i|^2+|l_i|^2 =1$. The above relation  is determined by the relation of particles to detectors (spatial modes), which can be quantified as the spatial coherence~\cite{chin2019entanglement}. 
Now the two identical particles that is actually detected are in the form
\begin{align}
     r_1r_2|R,s_1\>\otimes_{\pm}|R,s_2\> +  +l_1l_2|L,s_1\>\otimes_{\pm}|L,s_2\> 
    +  r_1l_2|R,s_1\>\otimes_{\pm}|L,s_2\> +  l_1r_2|L,s_1\>\otimes_{\pm}|R,s_2\>,
\end{align}
which is an entangled state at the level of detectors (or modes)~\cite{franco2018indistinguishability,chin2019entanglement,tichy2013entanglement, barros2019entangling}. This discussion shows that the nonlocality of identical particles cannot be read off just by looking into the wave functions. 
Particle identity and spatial coherence combine to generate genuine entanglement, and the final entanglement is obtained in the form of mode entanglement. 


\section{NSSR-preserving entanglement of bosons}\label{NSSR}

Suppose that there exist two systems $X$ and $Y$ that locate far from each other and have never exchanged any information, therefore separated. Over $X$ spread $n$ identical bosons and over $Y$ spread $(N-n)$ identical bosons. Each boson has an internal degree of freedom $s_i$ with $i=1,\cdots, S$.  
Then, in the SEA formalism, a separable $N$ boson state is written in the most general form as
\begin{align}\label{bisep.}
|\P_N^{sep}\> =& (\sum_a\p_{X}^a|X,s^a_1\> \vee \cdots \vee |X,s^a_n\>)\vee (\sum_b\p_Y^b|Y,s^b_{n+1}\> \vee \cdots \vee |Y,s^b_N\>) \nn \\
\equiv& |\P^X_n\> \vee |\P^Y_{N-n}\>, 
\end{align}
 where $\p_X^a$ and $\p_Y^b$ are complex numbers for the wave function normalization.
One can see that $|\P_N^{sep}\>$ is separable with respect to the systems $X$  and $Y$, because $|\P^X_n\>$ and $|\P^Y_{N-n}\>$ can be prepared in each system independently. Indeed, using Definition 2, $|\P_N^{sep}\>$ is prepared with creation operators as
\begin{align}\label{bosonsep}
 |\P_N^{sep}\> 
 &= \Big(\sum_{a}\p_X^a\ha^\dagger(X,s_1^a)\cdots \ha^\dagger(X,s_n^a)\Big) \Big(\sum_b\p_Y^b\ha^\dagger(Y,s_{n+1}^b)\cdots \ha^\dagger(Y,s_N^b)\Big)|vac\> \nn \\
 &\equiv  \ha^\dagger \big(\P^X_n\big)\ha^\dagger\big(\P^Y_{N-n}\big)|vac\>,
\end{align} 
 where  $[\ha^\dagger \big(\P^X_n\big), \ha^\dagger\big(\P^Y_{N-n}\big)]=0$.

\begin{figure}[t]
	\centering
	\includegraphics[width=8.5cm]{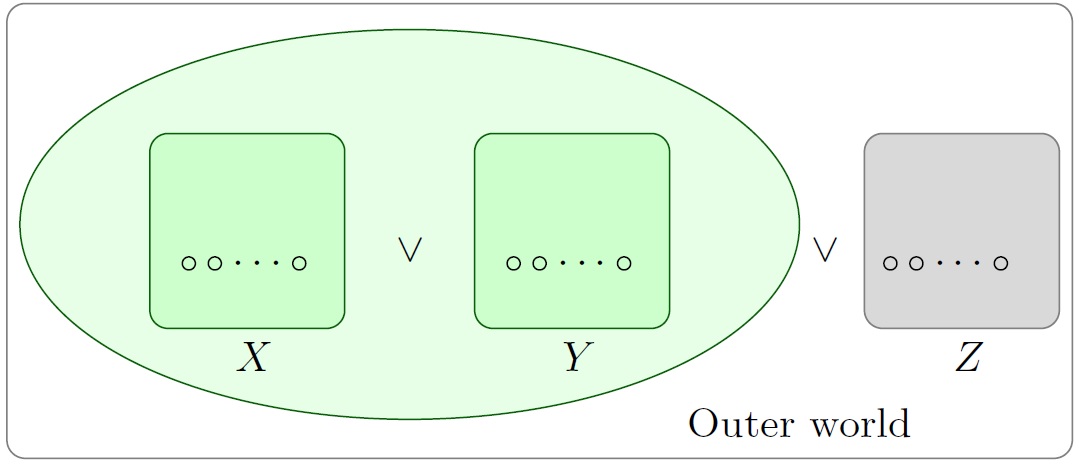}
\caption{An example of a tripartite system with bosons. While the subsystems $X$ and $Y$ have the potential to be correlated to each other, $Z$ is insulated to the other subsystems (included in the outer world). In this case, we only need to examine the nonlocality between two subsystems $X$ and $Y$, and the total system is considered $X\cup Y$. The relation of $X \cup Y$ with $Z$ is represented as the symmetric product $\vee$, which we do not care about as long as no history of physical interaction between $X\cup Y$ and $Z$ exists.}
\label{tripartite}
\end{figure}

This expression is powerful when we consider the third system that has $l$ identical bosons in a state $|\P^Z_l\>$ with no interaction to $X$ and $Y$ (Fig.~\ref{tripartite}). Even if the total state then can be rewritten as $|\P^X_{n}\>\vee |\P^Y_{N-n}\>\vee |\P^Z_l\>$,  two communicators in $X$ and $Y$ do not need to take $|\P^Z_{l}\>$ into account to evaluate the non-locality of them. Hence, with this $\vee$ notation (or $\wedge$ notation for fermions) one can treat any multipartite system of identical particles similar to the distinguishable particle case.

By extending the above discussion, the most general statement for the separable states of $N$ boson in $P$ subsystems is possible. A set of $N$ bosons that spreads over $P$ subsystems $X_i$ ($i=1,\cdots, P$) are separable if and only if the total state is given by
\begin{align}\label{separability_boson}
    |\P\> = \vee_{i=1}^{P} (\sum_{a_i}\p_{X_i}^{a_i}|X_i, s^{a_i}_1\>\vee\cdots \vee |X_i,s^{a_i}_{n_i}\>)
\end{align} where $\sum_{i=1}^P n_i = N$.
Our separability condition can be considered the generalization of that introduced in Ref.~\cite{wiseman2003entanglement}.

Now we  apply the bipartite separability condition Eq.~\eqref{bisep.} to the case when $N$ boson spread over space including two distinguishable detectors $X$ and $Y$ (Fig.~\ref{passiveop}), 
\begin{align}
    |\P\> =& |\P_1\>\vee |\P_2\>\vee \cdots\vee |\P_N\> \nn \\
    =& \vee_{i=1}^N(r_i|X,s_i\>+l_i|Y,s_i\>)
\end{align} where $\P_i =(\p_i,s_i) = (r_iX+l_iY,s_i )$. This state is separable when it can be written as
\begin{align}
    |\P\> = \sum_{n=0}^N |\P_n^X\>\vee|\P_{N-n}^Y\>
\end{align} ($|\P_n^X\>$  is a $n$-boson state  in $X$ and $|\P_(N-n)^Y\>$ is a $N-n$ boson state in $Y$), for $|\P_n^X\>\vee|\P_{N-n}^Y\>$ with different $n$ cannot superpose with each other. 

It is quite straightforward to define several entanglement measures for bipartite bosonic states that vanish when the states are separable. 
Here, we present the definition of entanglement entropy as an example.\\

\noindent
\paragraph*{Entanglement entropy.} The  entropy of a bipartite system that consists of $X$ and $Y$ can be defined with the \emph{symmetrized partial trace technic}~\cite{franco2016quantum,chin2019entanglement,chin2019reduced}. Suppose that a subsystem $X$ with $n$ bosons has a complete orthonormal basis set $\{|X,s_1^a\>\vee\cdots\vee|X, s_{n}^a\>\}_{a} \} \equiv \{|(s_1,\cdots,s_n)^{a}\>^X\}_a$ and the identity matrix is given by $\bI_X= \sum_{a}|(s_1,\cdots,  s_n)^a\>\<(s_1,\cdots,  s_n)^a|^X$. Then the reduced density matrix $\r_Y^n$ of $Y$ with respect to a total state $|\P\>$ is derived from Definition~\ref{partialsea} as
\begin{align}\label{reduced}
 \r_Y^{n} =& \Tr_X(\bI_X|\P\>\<\P|) \nn \\
 =& \sum_a \ha\Big((s_1,\cdots,  s_n)^a\Big)|\P\>\<\P|\ha^\dagger \Big((s_1,\cdots,  s_n)^a\Big)
\end{align}
and the entropy is given by 
\begin{align}
    E(|\P\>) =\sum_{n=1}^{N-1}P(\r_Y^n)E(\r_Y^{n}) = 
     -\sum_{n=1}^{N-1}P(\r_Y^n)\Tr(\r_Y^{n}\ln\r_Y^{n})
\end{align} where $P(\r_Y^n)$ is the probability for $\r_Y^n$ to be observed. 
Ref.~\cite{chin2019reduced} connects the symmetrized partial trace to the subalgebra restriction~\cite{balachandran2013entanglement,balachandran2013algebraic} in algebraic quantum mechanics.


\begin{figure}[t]
    \centering
    \includegraphics[width=6.2cm]{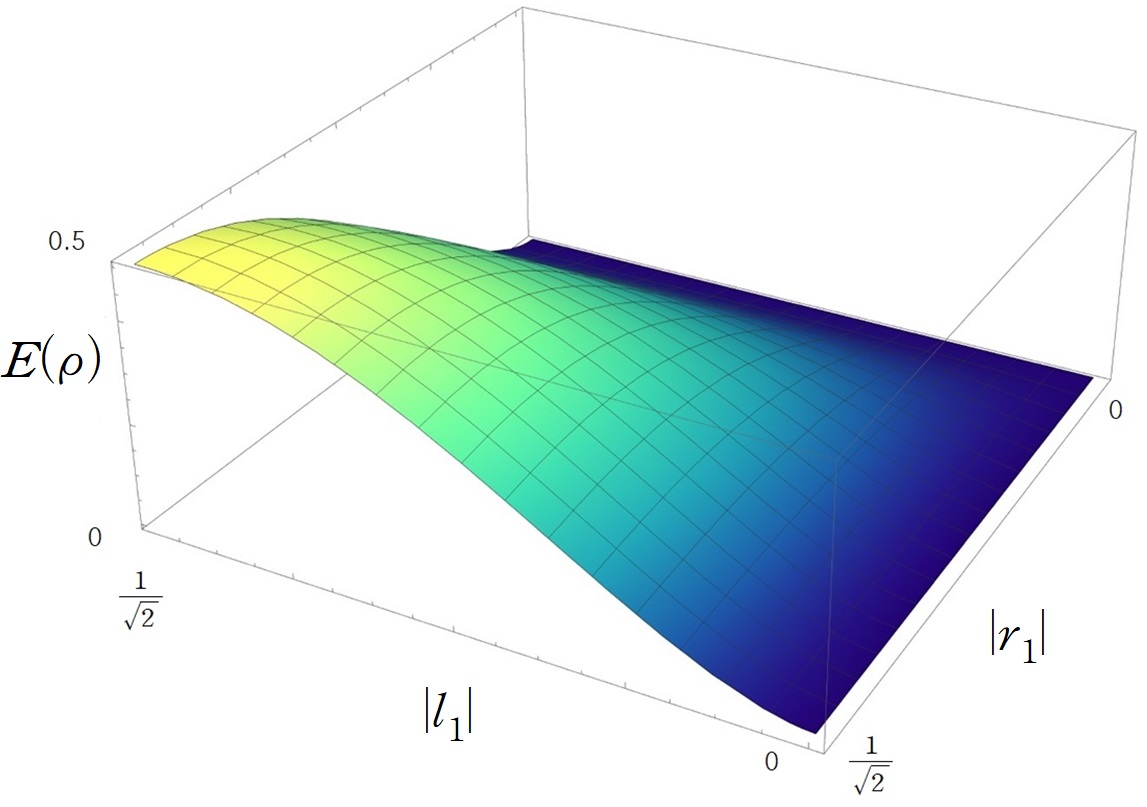}
    \caption{Entanglement entropy of $2$ bosons according to the variation of spatial coherence. The entropy is zero when one of $(r_1,r_2,l_1,l_2)$ is zero. The maximal $E(\r_{ent})$ is given when $|r_1|=|r_2|=|l_1|=|l_2|=\frac{1}{\sqrt{2}}$ by $1/2$. NSSR restricts the possible states to superpose to each other, which diminishes the maximal entropy to $1/2$.}
    \label{entropybs}
\end{figure}

As a simple example, we compute the entropy of two bosons with internal states $\uparrow$ and $\downarrow$ respectively. Then from Eq.~\eqref{separability_boson}, the state is given by
\begin{align}
 |\P\> =   \big( r_1r_2|\uparrow,\downarrow\>^X +l_1l_2|\uparrow,\downarrow\>^Y \big)+ \big( r_1l_2|\uparrow\>^X\vee|\downarrow\>^Y +l_1r_2|\downarrow\>^X\vee|\uparrow\>^Y\big).
\end{align} 
Considering the case when each detector observes one particle, $\bI_X$ is given by $\bI_X =\sum_{r,s=\uparrow,\downarrow} |r,s\>\<r,s|^X$ and  the reduced density matrix $\r_Y$ becomes
 \begin{align}
 \r_Y =& \frac{1}{|r_1l_2|^2+|r_2l_1|^2}\big(|r_1l_2|^2 |\downarrow\>\<\downarrow|^Y + |l_1r_2|^2 |\uparrow\>\<\uparrow|^Y \big)
\end{align} with probability $(|r_1l_2|^2 + |l_1r_2|^2)$. 
Hence, the entanglement entropy for $|\P\>$ is given by
\begin{align}\label{bsent}
    E(|\P\>) 
    =& (|r_1l_2|^2 + |l_1r_2|^2) E(\r_Y) \nn \\
    =&-|r_1l_2|^2\log\Big[\frac{|r_1l_2|^2}{|r_1l_2|^2 + |l_1r_2|^2}\Big]  -|l_1r_2|^2\log\Big[\frac{|l_1r_2|^2}{|r_1l_2|^2 + |l_1r_2|^2}\Big].
\end{align}
The state is unentangled when one of $(r_1,r_2,l_1,l_2)$ is zero. The maximal $E(\r_{ent})$ is given when $|r_1|=|r_2|=|l_1|=|l_2|=\frac{1}{\sqrt{2}}$ by $1/2$ (Fig.~\ref{entropybs}). The derivation of Eq.~\eqref{bsent} is given in the former works, e.g., Ref~\cite{franco2018indistinguishability,chin2019entanglement}, however we here reproduce it for the comparison with the fermionic case in Section~\ref{ferm}.

\bibliographystyle{unsrt}
\bibliography{Taming}

\end{document}